\documentclass[letterpaper, conference]{ieeeconf}  % Comment this line out

\usepackage{amsmath}
\usepackage{graphicx}
\usepackage{amssymb}
\usepackage{colortbl}
\usepackage{arydshln}
\usepackage{stfloats}
\usepackage{algorithm}
\usepackage{algorithmicx}
\usepackage{algpseudocode}
\usepackage{multirow}
\usepackage{extarrows}
\usepackage{mathrsfs}
\usepackage[sort]{cite}

\definecolor{mygray}{gray}{0.8}

\newtheorem{theorem}{Theorem}

\newtheorem{proposition}{Proposition}
\newtheorem{definition}{Definition}

\newtheorem{assumption}{Assumption}
\newtheorem{problem}{Problem}
\newtheorem{corollary}{Corollary}
\IEEEoverridecommandlockouts                              % This command is only
\title{\LARGE \bf
Reachable Set Estimation and Safety Verification for Piecewise Linear Systems with Neural Network Controllers 
}

\author{Weiming Xiang, Hoang-Dung Tran, Joel A. Rosenfeld, Taylor T. Johnson% <-this % stops a space
\thanks{The material presented in this paper is based upon work supported by
the National Science Foundation (NSF) under grant numbers CNS 1464311
and 1713253, and SHF 1527398 and 1736323, and the Air Force Office of
Scientific Research (AFOSR) under contract numbers FA9550-15-1-0258,
FA9550-16-1-0246, and FA9550-18-1-0122. The U.S. government is
authorized to reproduce and distribute reprints for Governmental
purposes notwithstanding any copyright notation thereon. Any opinions,
findings, and conclusions or recommendations expressed in this
publication are those of the authors and do not necessarily reflect
the views of AFOSR or NSF.}% <-this % stops a space
\thanks{Authors are with the Department of Electrical Engineering and Computer Science, Vanderbilt University, Nashville, Tennessee 37212, USA.
        {\tt\small xiangwming@gmail.com; trhoangdung@gmail.com;
        	joel.rosenfeld@vanderbilt.edu; taylor.johnson@gmail.com}}%%
}

\begin{document}

\maketitle
\thispagestyle{empty}
\pagestyle{empty}

%%%%%%%%%%%%%%%%%%%%%%%%%%%%%%%%%%%%%%%%%%%%%%%%%%%%%%%%%%%%%%%%%%%%%%%%%%%%%%%%
\begin{abstract}
\boldmath
In this work, the reachable set estimation and safety verification problems for a class of piecewise linear systems equipped with neural network controllers are addressed. The neural network is considered to consist of Rectified Linear Unit (ReLU) activation functions. A layer-by-layer approach is developed for the output reachable set computation of ReLU neural networks. The computation is formulated in the form of a set of manipulations for a union of polytopes. Based on the output reachable set for neural network controllers, the output reachable set for a piecewise linear feedback control system can be estimated iteratively for a given finite-time interval. With the estimated output reachable set, the safety verification for piecewise linear systems with neural network controllers can be performed by checking the existence of intersections of unsafe regions and output reach set. A numerical example is presented to illustrate the effectiveness of our approach. 

\end{abstract}

%%%%%%%%%%%%%%%%%%%%%%%%%%%%%%%%%%%%%%%%%%%%%%%%%%%%%%%%%%%%%%%%%%%%%%%%%%%%%%%%
\section{Introduction}
Artificial neural networks have been widely used in machine learning systems, especially in control systems where the plant models are complex or even unavailable, e.g., \cite{hunt1992neural,ge1999adaptive}. %Neural network based controllers have been demonstrated to be effective at controlling complex systems. 
However, such controllers are confined to systems which comply with the lowest safety integrity, since the majority of neural networks are viewed as \emph{black box} lacking effective methods to predict all outputs and assure safety specifications for closed-loop systems. In a variety of applications to feedback control systems, there are safety-oriented restrictions such that the system states are not allowed to reach unsafe regions while under the control of a neural network based feedback controller. Neural networks can react in unexpected and incorrect ways to even slight perturbations of their inputs \cite{szegedy2013intriguing}, thus it could result in unsafe closed-loop systems even while under control of  well-trained neural network controllers. Hence, methods that are able to provide formal guarantees are in a great demand for verifying specifications or properties of systems involving neural network controllers. Even the verification of simple properties concerning neural networks have been demonstrated to be NP-complete problems  \cite{katz2017reluplex}.  Few results have been reported in the literature for verifying systems involving neural networks. In \cite{huang2017safety} \emph{Satisfiability Modulo Theory} (SMT) is utilized for the verification of feed-forward multi-layer neural networks. In \cite{pulina2010abstraction} an abstraction-refinement approach is developed for computing output reachable set of neural networks. In \cite{Xiang2017reachable_arxiv,katz2017reluplex}, a specific kind of activation functions called \emph{Rectified Linear Unit} (ReLU) is considered for the verification of neural networks. A simulation-based approach is developed in \cite{Xiang2017output_arxiv}, which turns the reachable set estimation problem into a neural network maximal sensitivity computation problem that is described in terms of a chain of convex optimization problems. Recently, Lyapunov functions were utilized for reachable set estimation of neural networks in \cite{xu2017reachable,zuo2014non}. 

Piecewise linear systems have emerged as an important subclass of hybrid systems and represent a very active area of current research in the field of control systems \cite{Xiang2017stability,Decarlo2000perspectives,Liberzon2012switching,Xiang2016necessary}. %A piecewise linear system is composed of a family of continuous or discrete time modes, described by differential or difference equations, respectively, along with a switching rule governing the activation of modes. 
The motivation for studying piecewise linear systems comes from the fact that piecewise linear systems can be effectively used to model many practical systems that are inherently multi-model in the sense that several dynamic subsystem models are required to describe their behaviors such as uncertain systems \cite{Xiang2018Parameter}. Stability analysis and stabilization are the main concerns for piecewise liner systems \cite{Zhang2016mode,Xiang2017robust,Lin2009stability}. Some recent results can be found for reachable set estimation for piecewise linear systems \cite{Xiang2017output,Chen2016estimation,xiang2017reachable}. In this paper, we will study the reachable set estimation and verification problems for a class of piecewise linear systems with neural network controllers. Since the neural network controller exists in the control loop, it is essential to compute or estimate the output reachable set of the neural network controller to facilitate the computation of the reachable set of the entire closed-loop system. For a class of ReLU neural networks, the output reachable set computation is converted into a set of polytope operations. Then, extensions to reachable set estimation for closed-loop systems are made and moreover, the safety verification is then reduced to check for empty intersections between reachable set and unsafe regions. 

The rest of this paper is organized as follows. The problem formulation and preliminaries are given in Section II. The main results of output reachable set estimations for ReLU neural networks and discrete-time piecewise linear feedback control systems with ReLU neural network controllers are presented in Section III. In  Section IV, a numerical example is provided to illustrate the results. Conclusions are given in Section V.

\section{Preliminaries and Problem Formulation}
In the paper, a class of discrete-time piecewise linear systems is considered in the following form
\begin{equation}\label{system}
x(k+1) = A_{\sigma(k)}x(k)+ B_{\sigma(k)}u(k)
\end{equation}
where $x(k) \in {\mathbb{R}^{n_x}}$  is state vector and $u(k) \in {\mathbb{R}^{n_u}}$ is the control input. The switching signal $\sigma$ is defined as $\sigma:\mathbb{N} \to \{1,\ldots,N\}$, where $N$ is the number of modes in the piecewise system. The switching instants are expressed by a sequence $\mathcal{S} \triangleq \left\{ {k_m} \right\}_{m=0}^{\infty}$, where $k_0$  denotes the initial time and $k_m$  denotes the $m$th switching instant.

Due to the presence of switching signal $\sigma$, piecewise linear system (\ref{system}) has much more complex behaviors than those are presented in linear systems. For the problem of controller design, our aim is to find  a feedback controller
\begin{equation}\label{controller}
u(k) = g(x(k))
\end{equation}
where $g: \mathbb{R}^{n_x} \to \mathbb{R}^{n_u}$ is a  static feedback controller. The corresponding closed-loop system becomes
\begin{equation}\label{closedloop}
x(k+1) = A_{\sigma(k)}x(k)+ B_{\sigma(k)}g(x(k)).
\end{equation}

It is noted that controller (\ref{controller}) includes the common linear feedback scheme $g(x(k)) = Kx(k)$ and the mode-dependent linear feedback controller $g(x(k)) = K_{\sigma(k)}x(k)$, which have be widely used in the literature. However, it still has a number of challenges for controller design problems of system (\ref{system}), especially when the switching signal is not available for the design process. For example, when the switching signal $\sigma$ is unavailable and a common feedback controller has to be designed, the resulting controller is usually designed to be overly conservative.

For general nonlinear systems which include system (\ref{system}), the neural network based design method is a promising approach to resolve controller design problems for complex systems. In this paper, we consider a class of feedforward neural networks  called the Multi-Layer Perceptron (MLP), which consists of a number of interconnected neurons. The action of a neuron depends on its activation function, which is described as 
\begin{align}
y_i = f\left(\sum\nolimits_{j=1}^{n}\omega_{ij} v_j + \theta_i\right)
\end{align}
where $v_j$ is the $j$th input of the $i$th neuron, $\omega_{ij}$ is the weight from the $j$th input to the $i$th neuron, $\theta_i$ is called the bias of the $i$th neuron, $y_i$ is the output of the $i$th neuron, $f(\cdot)$ is the activation function. The activation function is generally a nonlinear function  describing the reaction of $i$th neuron with inputs $v_j$, $j=1,\ldots,n$. Typical activation functions include rectified linear unit, logistic, tanh, exponential linear unit, linear functions. 

An MLP has multiple layers,  each layer $\ell$, $1 \le \ell \le L $, has $n^{[\ell]}$ neurons.  In particular, layer $\ell =0$ is used to denote the input layer and $n^{[0]}$ stands for the number of inputs for the neural network, and  $n^{[L]}$ is the number of neurons in the output layer. For a neuron $i$, $1 \le i \le n^{[\ell]}$ in layer $\ell$, the corresponding input vector is denoted by $v^{[\ell]}$ and the weight matrix is 
\begin{equation*}
W^{[\ell]} = \left[\omega_{1}^{[\ell]},\ldots,\omega_{n^{[\ell]}}^{[\ell]}\right]^{\top}
\end{equation*}
where $\omega_{i}^{[\ell]}$ is the weight vector. The bias vector for layer $\ell$ is
\begin{equation*} \theta^{[\ell]}=\left[\theta_1^{[\ell]},\ldots,\theta_{n^{[\ell]}}^{[\ell]}\right]^{\top}.
\end{equation*} 

The output vector of layer $\ell$ can be expressed as 
\begin{equation*}
y^{[\ell]}=f_{\ell}({W}^{[\ell]}v^{[\ell]}+\theta^{[\ell]})
\end{equation*} 
where $f_{\ell}(\cdot)$ is the activation function for layer $\ell$.

In this work, we consider a ReLU activation function expressed as:
\begin{equation} \label{relu}
f(v) = v^{+} = \max(0,v).
\end{equation}

The output of a neuron  is rewritten as
\begin{equation}  \label{reluLayer}
y_i = \max\left(0,\sum\nolimits_{j=1}^{n}\omega_{ij} v_j + \theta_i\right)
\end{equation}
and the corresponding output vector of layer $\ell$ becomes
\begin{equation} \label{layer_l}
y^{[\ell]}=\max(0,W^{[\ell]}v^{[\ell]}+{\theta}^{[\ell]}).
\end{equation}

For an MLP, the output of $\ell-1$ layer is the input of the $\ell$ layer, and the mapping from the input layer $v^{[0]}$ to the output of output layer $y^{[L]}$ is the input-output relation of the MLP, denoted by
\begin{equation}\label{NN}
y^{[L]} = g (v^{[0]})
\end{equation}    
where $g(\cdot) \triangleq f_L  \circ f_{L - 1}  \circ  \cdots  \circ f_1(\cdot) $. 

For controller design problem, we can let input of MLP $v^{[0]}=x(k)$ and output of MLP $y^{[L]} = u(k)$. That means, given a feedback control system of the form (2) and a control objective, a neural network can be trained to achieve the control objective. There are a variety of results for designing neural network based feedback controller $g(x(k))$.
Despite a neural network's ability to approximate nonlinear functions through the universality property, predicting the output behaviors of MLPs given in (\ref{NN})  still poses a significant challenge due to the nonlinearity and nonconvexity of MLPs.  An MLP is usually viewed as a \emph{black box} to generate  a desirable output with respect to a given input. However, when considering  property verification which includes safety verification, it has been observed that even a well-trained neural network can react in unexpected and incorrect manners to  slight perturbations of their inputs, which could result in unsafe systems.  Thus, the output reachable set computation or estimation of an MLP, which encompasses all possible values of outputs, is necessary to verify the safety property of a neural network based feedback control system. 

Given an initial set $\mathcal{X}_0$, the reachable set of system (\ref{closedloop}) defined at time $k$ and over an interval $[0,k]$ is given by the following definition.
\begin{definition} \label{reachable_set}
	Given a piecewise linear system (\ref{system}) with a neural network controller (\ref{controller}) and initial state $x(0)$ belonging to a set $\mathcal{X}_0$, the reachable set of closed-loop system (\ref{closedloop}) at time $k$ is defined as 
	\begin{equation}\label{def_1}
	\mathcal{X}_k \triangleq \{x(k)  \mid x(k)~\mathrm{satisfies}~ (\ref{closedloop})~\mathrm{and}~x(0) \in \mathcal{X}_0\}
	\end{equation}
	and the reachable set over time interval $[0,k]$ is defined by
	\begin{equation}\label{def_2}
	\mathcal{X}_{[0,k]} = \bigcup\nolimits_{h=0}^{k}\mathcal{X}_{h}
	\end{equation} 
	where $\mathcal{X}_h$, $h=0,1,\ldots,k$, are defined in (\ref{def_1})
\end{definition}

In this paper, two problems will be addressed for feedback system (\ref{closedloop}) equipped with a neural network controller of the form (\ref{controller}). 

\begin{problem} \label{problem1}
		How does one compute the sets $\mathcal{X}_k$ and $\mathcal{X}_{[0,k]}$ given a piecewise linear system of the form (\ref{system}) with a neural network controller (\ref{controller}) and initial state $x(0)$ belonging to set $\mathcal{X}_0$? 
\end{problem}

 The safety specification for a closed-loop system of the form (\ref{closedloop}) is expressed by a set defined in the state space, describing the safety requirement. 

\begin{definition}
	A safety specification $\mathcal{S}$ 
	formalizes the safety requirements for a closed-loop system of the form (\ref{closedloop}) and is a predicate over the system state $x$ of the closed-loop system. The  closed-loop system (\ref{closedloop}) is safe over the interval $[0,k]$ if and only if the following condition is satisfied:
	\begin{equation}\label{safety}
	\mathcal{X}_{[0,k]} \cap \neg \mathcal{S} = \emptyset
	\end{equation}
	where $\neg$ is the symbol for logical negation. 
\end{definition}

The safety verification problem for closed-loop system (\ref{closedloop}) is stated as follows.
\begin{problem}\label{problem2}
		How can the safety requirement in (\ref{safety}) be verified given a piecewise linear system of the form (\ref{system}) with a neural network controller (\ref{controller}),  initial state $x(0)$ belonging to a set $\mathcal{X}_0$ and a safety specification $\mathcal{S}$? 
\end{problem}

To facilitate the developments in this paper, the following assumption is made.
\begin{assumption}
	 Initial state set $\mathcal{X}_0$ is considered to be a union of $N_0$ 
	 polytopes, that is expressed as $\mathcal{X}_0 = \bigcup_{s =1}^{N_0}{\mathcal{X}_{s,0}}$, where $\mathcal{X}_{s,0}$, $s = 1,\ldots,N_0$, are described by
	 \begin{equation}
	 \mathcal{X}_{s,0} \triangleq\left\{ x \mid  H_{s,0}x \le b_{s,0},x \in \mathbb{R}^{n_x}\right\},~s =1,\ldots,N_0. \label{inputset}
	 \end{equation}
\end{assumption}

In the following section, the main results on reachable set estimation and verification for piecewise linear systems with neural network controllers will be presented.

\section{Main Results}
	In this section, a layer-by-layer method is developed for computing the output reachable set for a ReLU neural network. We consider a single layer with ReLU neurons  and an indicator vector $q = [q_0,\ldots,q_{n}]$, $q_i \in \{0,1\}$ is utilized, in which the element $q_i$ is valuated as below:
	\begin{equation}
	q_i  = \left\{ {\begin{array}{*{20}c}
		0 & { \sum\nolimits_{j=1}^{n}\omega_{ij}v_j+\theta_i \le 0}  \\
		1 & {\sum\nolimits_{j=1}^{n}\omega_{ij}v_j+\theta_i > 0}  \\	
		\end{array} }. \right.
	\end{equation}
	
	There are $2^n$ possible indicator vectors $q$ in total, which are indexed as $q_0 = [0,0,\ldots,0]$, $q_1 = [0,0,\ldots,1]$, $\ldots$, $q_{2^{n-1}} = [1,1,\ldots,1]$. In the sequel, all these vectors from $q_0$ to $q_{2^n-1}$ are diagonalized as $Q_0 = \mathrm{diag}(q_0)$, $Q_1 = \mathrm{diag}(q_1)$, $\ldots$,  $Q_{2^{n-1}}= \mathrm{diag}(q_{2^{n-1}})$. 	
	
\begin{theorem}\label{thm1}
	Given a single layer described by (\ref{layer_l}) and an input set  $\mathcal{V}^{[\ell]}$ for the layer, the output set is 
	\begin{equation}
\mathcal{Y}^{[\ell]} = \bar{\mathcal{Y}}^{[\ell]} \cup \hat{\mathcal{Y}}^{[\ell]} \cup \left(\bigcup\nolimits_{m=1}^{2^{n^{[\ell]}}-2}\mathcal{Y}_{m}^{[\ell]}\right)
	\end{equation}
	where $\bar{\mathcal{Y}}^{[\ell]}$, $\hat{\mathcal{Y}}^{[\ell]}$, $\mathcal{Y}_{m}^{[\ell]}$, are defined as below
	\begin{align*}
	\bar{\mathcal{Y}}^{[\ell]} &= \left\{ {\begin{array}{*{20}c}
		{\{ 0\}, } & {\bar{\mathcal{V}}^ {[\ell]}   \ne \emptyset }  \\
		\emptyset,  & {\bar{\mathcal{V}}^ {[\ell]} = \emptyset }  \\	
		\end{array} } \right.,
	\\
	\bar{\mathcal{V}}^{[\ell]}  &= \{v \mid W^{[\ell]}v + \theta^{[\ell]} \le 0,~v \in \mathcal{V}^{[\ell]}\};
	\\
    \hat{\mathcal{Y}}^{[\ell]} &= \{y = W^{[\ell]}v+\theta^{[\ell]}\mid v\in \hat{\mathcal{V}}^{[\ell]}\},
    \\    
    \hat{\mathcal{V}}^{[\ell]} & = \{v \mid W^{[\ell]}v + \theta^{[\ell]} > 0,v \in \mathcal{V}^{[\ell]}\};
    \\
    \mathcal{Y}_{m}^{[\ell]} &= \{ y = W^{[\ell]}v+\theta^{[\ell]}\mid v \in \bar{\mathcal{V}}_{m}^{[\ell]} \cap \hat{\mathcal{V}}_{m}^{[\ell]}\},
    \\\bar{\mathcal{V}}_{m}^{[\ell]} &= \{v \mid (I-Q_m)(W^{[\ell]}v+\theta^{[\ell]}) \le 0 ,v \in \mathcal{V}^{[\ell]}\},
    \\
    \hat{\mathcal{V}}_{m}^{[\ell]}&=\{v \mid Q_m(W^{[\ell]}v+\theta^{[\ell]}) \ge 0, v \in \mathcal{V}^{[\ell]}\}.
	\end{align*}
\end{theorem}
\begin{proof}
	For the inputs of the layer as $v^{[\ell]} \in \mathcal{V}^{[\ell]}$, we have three cases listed below to completely characterize the outputs of layer (\ref{layer_l}).
	
    \emph{Case 1}:   All the elements in the outputs are non-positive, which means
		\begin{equation}
		v^{[\ell]} \in \bar{\mathcal{V}}^{[\ell]}  = \{v \mid W^{[\ell]}v + \theta^{[\ell]} \le 0,~v \in \mathcal{V}^{[\ell]}\}.
		\end{equation}
        
By the definition of ReLU, it directly yields the output set for this case is 
		\begin{equation}
		\bar{\mathcal{Y}}^{[\ell]}= \left\{ {\begin{array}{*{20}c}
			{\{ 0\}, } & {\bar{\mathcal{V}}^ {[\ell]}   \ne \emptyset }  \\
			\emptyset,  & {\bar{\mathcal{V}}^ {[\ell]} = \emptyset }  \\	
			\end{array} } \right..
		\end{equation} 
		
		\emph{Case 2}: All the elements are   positive by the input in $\mathcal{V}^{[\ell]}$, that implies
		\begin{equation}
		v^{[\ell]}  \in \hat{\mathcal{V}}^{[\ell]} = \{v \mid W^{[\ell]}v + \theta^{[\ell]} > 0,v \in \mathcal{V}^{[\ell]}\}.
		\end{equation}
	
    So, the output set is 
		\begin{equation}
		\hat{\mathcal{Y}}^{[\ell]} = \{y \mid y = W^{[\ell]}v+\theta^{[\ell]},v\in \hat{\mathcal{V}}^{[\ell]}\}.
		\end{equation}
		
		\emph{Case 3}: Outputs have both negative and positive elements, which correspond to indicator vectors $q_m$, $m=1,\ldots,2^{n^{[\ell]}}-2$.  Note that, for each $q_m$, $m=1,\ldots,2^{n^{[\ell]}}-2$, the element $q_i = 0$ indicates $y_i=\max(0,\sum\nolimits_{j=1}^{n^{[\ell]}}\omega_{ij}v_j+\theta_i)=  0$ due to $\sum\nolimits_{j=1}^{n^{[\ell]}}\omega_{ij}v_j+\theta_i \le 0$. With respect to each $q_m$, $m=1,\ldots,2^{n^{[\ell]}}-2$, we define  set 
		\begin{equation}
		\bar{\mathcal{V}}_{m}^{[\ell]} = \{v\mid \sum\nolimits_{j=1}^{n^{[\ell]}}\omega_{ij}^{[\ell]}v_j+\theta_i^{[\ell]} \le 0,v \in \mathcal{V}^{[\ell]}\}
		\end{equation}
		in which $i \in \{i \mid q_i =0~\mathrm{in}~q_m =[q_1,\ldots,q_{n^{[\ell]}]}\}$. In a compact form, it can be expressed as
		\begin{equation}
		\bar{\mathcal{V}}_{m}^{[\ell]} = \{v \mid (I-Q_m)(W^{[\ell]}v+\theta^{[\ell]}) \le 0 ,v \in \mathcal{V}^{[\ell]}\}.
		\end{equation}
		
		Due to  ReLU functions, when $\sum\nolimits_{j=1}^{n^{[\ell]}}\omega_{ij}v_j+\theta_i \le 0$, it will be set to 0, thus the output for $v \in \bar{\mathcal{V}}_{m}$ should be
		\begin{equation}
		\bar{\mathcal{Y}}_{m}^{[\ell]} = \{ y=Q_m(W^{[\ell]}v+\theta^{[\ell]}) \mid v \in \bar{\mathcal{V}}_{m}^{[\ell]}\}.
		\end{equation}
        
		Again, due to ReLU functions, the final value should be non-negative, that is $y \ge 0$, thus additional  constraint for $v^{[\ell]}$ has to be added as 
		\begin{equation}
		v^{[\ell]} \in \hat{\mathcal{V}}_{m}^{[\ell]}=\{v \mid Q_m(W^{[\ell]}v+\theta^{[\ell]}) \ge 0, v \in \bar{\mathcal{V}}_{m}^{[\ell]}\}.
		\end{equation}
        
 The resulting output set is \begin{equation}\mathcal{Y}_{m}^{[\ell]} = \{ y = W^{[\ell]}v+\theta^{[\ell]} \mid v \in \bar{\mathcal{V}}_{m}^{[\ell]} \cap \hat{\mathcal{V}}_{m}^{[\ell]}\}.
		\end{equation}
	
	The three cases establish that the output set generated from input set $\mathcal{V}^{[\ell]}$ is 
	\begin{equation}\label{ysl}
	\mathcal{Y}^{[\ell]} = \bar{\mathcal{Y}}^{[\ell]} \cup \hat{\mathcal{Y}}^{[\ell]} \cup \left(\bigcup\nolimits_{m=1}^{2^{n^{[\ell]}}-2}\mathcal{Y}_{m}^{[\ell]}\right).
	\end{equation}
The proof is complete.
\end{proof}

From Theorem \ref{thm1}, the following corollary can be obtained if the input set for the layer is a union of polytopes. 
\begin{corollary}\label{cor1}
		Consider a signal layer (\ref{layer_l}), if the input set $\mathcal{V}^{[\ell]}$ is $ \mathcal{V}^{[\ell]} =  \bigcup_{s=1}^{N_{\ell}}\mathcal{V}_s^{[\ell]}$, where $\mathcal{V}^{[\ell]}_s$, $s \in \{1,\ldots,N_{\ell}\}$, are polytopes described as 
		\begin{equation}\label{cor_1}
		\mathcal{V}^{[\ell]}_s\triangleq \{v \mid H_s^{[\ell]}v\le b_s^{[\ell]}, v \in \mathbb{R}^{n^{[\ell]}}\}
		\end{equation}
		 then the output set $\mathcal{Y}^{[\ell]}$ is also a union of polytopes.
\end{corollary}
\begin{proof}
	By Theorem \ref{thm1},  $\bar{\mathcal{Y}}^{[\ell]}$, $\hat{\mathcal{Y}}^{[\ell]}$, $\mathcal{Y}_{m}^{[\ell]}$, are polytopes if $\mathcal{V}^{[\ell]}$ is a polytope. Thus, for $\mathcal{V}_s^{[\ell]}$ in (\ref{cor_1}), the corresponding output set, $\mathcal{Y}^{[\ell]}_s$, is a union of polytopes. Moreover, for $\mathcal{V}^{[\ell]} =  \bigcup_{s=1}^{N_{\ell}}\mathcal{V}_s^{[\ell]}$, the output set is
	\begin{equation}
	\mathcal{Y}^{[\ell]} =\bigcup\nolimits_{s=1}^{N_\ell}\mathcal{Y}^{[\ell]}_s
	\end{equation}
	which is  a union of polytopes.  
\end{proof}

By Corollary \ref{cor1}, if input set $\mathcal{V}^{[\ell]}$ is given as (\ref{cor_1}), the output can be expressed by a union of polytopes  by Theorem \ref{thm1}. The set $\bar{\mathcal{Y}}^{[\ell]}_s$, $\hat{\mathcal{Y}}^{[\ell]}_s$, $\mathcal{Y}_{s,m}^{[\ell]}$ can be expressed as follows:

\begin{align}
\bar{\mathcal{Y}}_s^{[\ell]}  & = \left\{ {\begin{array}{*{20}c}
	{\{ 0\}, } & {\bar{\mathcal{V}}_s^ {[\ell]}   \ne \emptyset }  \\
	\emptyset,  & {\bar{\mathcal{V}}_s^ {[\ell]}  = \emptyset }  \\	
	\end{array} } \right., \label{barY}
\\
\bar{\mathcal{V}}_s^{[\ell]}& = \left\{v \mid \left[ {\begin{array}{*{20}c}
	{H_s^{[\ell]} }  \\
	W^{[\ell] }  \\
	\end{array} } \right]v \le \left[ {\begin{array}{*{20}c}
	{b_s^{[\ell]} }  \\
	\theta^{[\ell]}  \\
	\end{array} } \right]\right\}; \label{barV}
\\
\hat{\mathcal{Y}}_s^{[\ell]}& = \left\{ y =W^{[\ell]}v+\theta^{[\ell]} \mid v \in \hat{\mathcal{V}}_s^{[\ell]}\right\}, \label{hatY}
\\
\hat{\mathcal{V}}_s^{[\ell]} &= \left\{v\mid H_s^{[\ell]} v\le b_s^{[\ell]} \wedge W^{[\ell]}v> -\theta^{[\ell]} \right\}; \label{hatV}
\\
\mathcal{Y}_{s,m}^{[\ell]}& = \left\{y =W^{[\ell]}v+\theta^{[\ell]} \mid v \in \mathcal{V}_{m,s}^{[\ell]}\right\}, \label{Ysm}
\\
\mathcal{V}_{m,s}^{[\ell]} &=\left\{v\mid H_{m,s}^{[\ell]}v \le b_{m,s}^{[\ell]}\right\}, \label{Vsm}
\\
H_{m,s}^{[\ell]}&=\left[ {\begin{array}{*{20}c}
	{H_s^{[\ell]} }  \\
	{(I - Q_m)W^{[\ell]}}  \\
	{ - Q_m W^{[\ell]}}  \\
	\end{array} } \right],b_{m,s}^{[\ell]} = \left[ {\begin{array}{*{20}c}
	{b_s^{[\ell]} }  \\
	( Q_m-I)\theta^{[\ell]}  \\
	Q_m\theta^{[\ell]}  \\	
	\end{array} } \right].\label{Hbsm}
\end{align}
where $\wedge$ is the symbol of logical conjunction. 

The  algorithm for generating the output set of layer $\ell$ is summarized in Algorithm \ref{algorithm_1}.

\begin{algorithm}
	\caption{Output Reach Set Computation for ReLU Layers} \label{algorithm_1}
	
	\begin{algorithmic}[1]
		\Require Neural network weight matrix $W^{[\ell]}$ and bias $~\theta^{[\ell]}$, input set $\mathcal{V}^{[\ell]}= \bigcup_{s=1}^{N_{\ell}}\mathcal{V}_s^{[\ell]}$ with $\mathcal{V}^{[\ell]}_s \triangleq\{v \mid  H_s^{[\ell]}v\le b_s^{[\ell]}\} $.
		\Ensure Output reach set $\mathcal{Y}^{[\ell]}$.
		
		\Function{layeroutput}{$W^{[\ell]},~\theta^{[\ell]}$,~$\mathcal{V}^{[\ell]}$}
		\For{$s=1:1:N_{\ell}$}
		\State Compute $\bar{\mathcal{V}}_s^{[\ell]}$ by (\ref{barV})
				
		\If{$\bar{\mathcal{V}}_s^{[\ell]} \ne \emptyset$}
		\State $\bar{\mathcal{Y}}_s^{[\ell]} \gets \{0\} $
		\Else
		\State $\bar{\mathcal{Y}}_s^{[\ell]} \gets \emptyset$
		\EndIf
		
		\State Compute $\hat{\mathcal{Y}}_s^{[\ell]}$ by (\ref{hatY}), (\ref{hatV})
		
		\For{$h=1:1:2^n-2$}
		\State Compute $\mathcal{Y}_{s,m}$ by (\ref{Ysm})--(\ref{Hbsm})
		
		\EndFor
		\State $\mathcal{Y}^{[\ell]}_s \gets \bar{\mathcal{Y}}^{[\ell]}_s \cup \hat{\mathcal{Y}}^{[\ell]}_s \cup\left( \bigcup\nolimits_{m=1}^{2^{n^{[\ell]}}-2}\mathcal{Y}_{s,m}^{[\ell]}\right)$
		\EndFor	
		
		\State \Return $\mathcal{Y}^{[\ell]} \gets \bigcup\nolimits_{s=1}^{N_{\ell}}\mathcal{Y}^{[\ell]}_s$
		
		\EndFunction
	\end{algorithmic}
\end{algorithm}

As for linear activation functions, which are commonly used in the output layer $L$, the output reach set can be computed in a similar manner to the set $\hat{\mathcal{Y}}_s^{[\ell]}$ for ReLU, this time omitting the constraint $y > 0$. 

\begin{corollary}
	Consider a linear layer $y^{[\ell]}=W^{[\ell]}v^{[\ell]}+\theta^{[\ell]}$ with input set $\mathcal{V}^{[\ell]}$, the output reach set $\mathcal{Y}^{[\ell]}$ of linear layer $\ell$ is
	\begin{equation}\label{linearY}
\mathcal{Y}^{[\ell]} = \left\{y =W^{[\ell]}v+\theta^{[\ell]} \mid v \in \mathcal{V}^{[\ell]}\right\}.
	\end{equation}
Moreover, if the input set $\mathcal{V}^{[\ell]}$ is a union of polytopes, the output set $\mathcal{Y}^{[\ell]}$ is still a union of polytopes.
\end{corollary}
\begin{proof}
  Since linear layer has $y^{[\ell]}=W^{[\ell]}v^{[\ell]}+\theta^{[\ell]}$, the output set $\mathcal{Y}^{[\ell]}$ is expressed by (\ref{linearY}). Furthermore, by the similar proof line in Corollary \ref{cor1}, $\mathcal{Y}^{[\ell]}$ is a union of ploytopes if input set $\mathcal{V}^{[\ell]}$ is a union of polytopes.  
\end{proof}

 For a multi-layer neural network, it can be observed that $v^{[\ell]}=y^{[\ell-1]}$, $\ell=1,\ldots,L$, as well as $\mathcal{V}^{[\ell]} = \mathcal{Y}^{[\ell-1]}$,  $\ell=1,\ldots,L$. Hence, the neural network (\ref{NN}) can be expressed recursively as
 \begin{equation} \label{layer_l_iteration}
 y^{[\ell]}=\max(0,W^{[\ell]}y^{[\ell-1]}+\theta^{[\ell]}),~\ell =1,\ldots,L 
 \end{equation} 
 where $y^{[0]} = v^{[0]}$ is the input of the neural network and $y^{[L]}$ is the output of the neural network, respectively. Accordingly, the input set and output are denoted by $\mathcal{V}^{[0]}$ and $\mathcal{Y}^{[L]}$. 

\begin{proposition}\label{proposition_1}
Given an MLP with $L$ layers with corresponding ReLU or linear activation functions for each layer $\ell$, $\ell =1,\ldots, L$, and given the input set $\mathcal{V}^{[0]}$  as described by (\ref{cor_1}) , the output reach set $\mathcal{Y}^{[L]}$ can be computed recursively using (\ref{barY})--(\ref{linearY})  and Algorithm \ref{algorithm_1}.
\end{proposition}
\begin{proof}
	It can be derived  directly using $\mathcal{V}^{[\ell]} = \mathcal{Y}^{[\ell-1]}$. The proof is complete. 
\end{proof}

The routine for computing the output reach set produced by ReLU neural networks is outlined in Algorithm  \ref{algorithm_2}.
\begin{algorithm}
	\caption{Output Reach Set Computation for ReLU Networks} \label{algorithm_2}
	
	\begin{algorithmic}[1]
		\Require Neural network weight matrix $W^{[\ell]}$ and bias $\theta^{[\ell]}$, $\ell=1,\ldots,L$, input set $\mathcal{V}^{[0]}= \bigcup_{s=1}^{N_{0}}\mathcal{V}_s^{[0]}$ with $\mathcal{V}^{[0]}_s \triangleq\{v \mid  H_s^{[0]}v \le b_s^{[0]}\} $.
		\Ensure Output reach set $\mathcal{Y}^{[L]}$.
		
		\Function{networkoutput}{$W^{[\ell]},~\theta^{[\ell]}$,~$\mathcal{V}^{[0]}$}
		\State $\mathcal{V}^{[1]} \gets \mathcal{V}^{[0]}$
		\For{$\ell=1:1:L$}
		\If{Layer $\ell$ is a ReLU Layer}

					\State $\mathcal{Y}^{[\ell]} \gets\mathrm{layerouput}(W^{[\ell]},~\theta^{[\ell]},\mathcal{V}^{[\ell]})$
	
		\ElsIf{Layer $\ell$ is a linear layer}
		
			\State Compute $\mathcal{Y}^{[\ell]}$ by (\ref{linearY})

		\EndIf
		
		\If{$\ell <L$}
		
		\State $\mathcal{V}^{[\ell+1]} \gets \mathcal{Y}^{[\ell]}$
		
		\ElsIf{$\ell =L$}
			\State \Return $\mathcal{Y}^{[L]}$
		\EndIf
		
		\EndFor
	
		\EndFunction
	\end{algorithmic}
\end{algorithm}

\begin{algorithm}
	\caption{Reachable Set Computation for piecewise linear systems with neural network controller} \label{algorithm_3}
	
	\begin{algorithmic}[1]
		\Require Neural network weight matrix $W^{[\ell]}$ and bias $\theta^{[\ell]}$, $\ell=1,\ldots,L$, system matrices $A_i$, $B_i$, $i=1,\ldots,N$, initial state set $\mathcal{X}_{0}= \bigcup_{s=1}^{N_{0}}\mathcal{X}_{s,0}$ with $ \mathcal{X}_{s,0} \triangleq\left\{ x \mid  H_{s,0}x \le b_{s,0},~x \in \mathbb{R}^{n}\right\}$.
		\Ensure Reachable set $\mathcal{X}_{k}$.
		
		\Function{systemreach}{$W^{[\ell]},~\theta^{[\ell]}$,~ $A_i$,~$B_i$, ~$\mathcal{X}_{0}$}
		\For{$h = 0:1:k-1$}
		\State $\mathcal{G}_h  \gets \mathrm{networkoutput}(W^{[\ell]},\theta^{[\ell]},\mathcal{X}_{h} )$
		\State Compute $\mathcal{X}_{h+1}$ by (\ref{Xh+1})
		\State $\mathcal{X}_{[0,h+1]} \gets \mathcal{X}_{h+1} \cup \mathcal{X}_{[0,h]}$
		\EndFor
		
		\State \Return$\mathcal{X}_{k}$ and $\mathcal{X}_{[0,k]}$
		
		\EndFunction
	\end{algorithmic}
\end{algorithm}

For piecewise linear systems as in  (\ref{system}), the reachable set $\mathcal{X}_k$ can be computed by the following proposition. 

\begin{proposition}\label{proposition_2}
	Consider piecewise linear system (\ref{system}) with initial state set $\mathcal{X}_0$ in (\ref{inputset}),  the reachable set $\mathcal{X}_k$ can be iteratively computed as 
	\begin{align}
	\mathcal{X}_{k+1}  = \{ x\mid x=A_{\sigma (k )}& x(k) + B_{\sigma (k )} g(x(k )), \nonumber
	\\
	~&~g(x(k)) \in \mathcal{G}_k,~x(k) \in \mathcal{X}_{k}\} \label{Xh+1}
	\end{align} 
	where $	\mathcal{G}_k  = \mathrm{networkoutput}(W^{[\ell]},\theta^{[\ell]},\mathcal{X}_{k} )$ is the output of Algorithm \ref{algorithm_2} and, the reachable set over time interval $[0,k]$ is
	\begin{equation}\label{prop_2}
\mathcal{X}_{[0,k]}=\bigcup\nolimits_{h=0}^{k}\mathcal{X}_h.
	\end{equation}
\end{proposition}
\begin{proof}
	By $x(k+1) = A_{\sigma (k )} x(k) + B_{\sigma (k )} g(x(k ))$ and $	\mathcal{G}_k  = \mathrm{networkoutput}(W^{[\ell]},\theta^{[\ell]},\mathcal{X}_{k} )$ in Proposition \ref{proposition_1}, it can obtained $\mathcal{X}_{k+1}$ by (\ref{Xh+1}). Then, according to Definition \ref{reachable_set}, $\mathcal{X}_{[0,k]}$ should be (\ref{prop_2}). 
\end{proof}

The algorithm for reachable set computation of a closed-loop system is summarized in Algorithm \ref{algorithm_3}.

%With the computed reachable set ${\mathcal{X}}_k$, the following result can be obtained for safety verification. 

\begin{proposition}
		Consider piecewise linear system (\ref{system}) with initial set (\ref{inputset}) and a given safety specification $\mathcal{S}$,  closed-loop system (\ref{closedloop}) is safe for interval $[0,k]$ if ${\mathcal{X}}_{[0,k]} \cap \neg \mathcal{S} = \emptyset$, where $\mathcal{X}_{[0,k]}$ is obtained by (\ref{prop_2}) in Proposition \ref{proposition_2}. 
\end{proposition}
\begin{proof}
	It is an immediate result by Definition \ref{reachable_set} and Proposition \ref{proposition_2}. The proof is complete. 
\end{proof}

The output of Algorithm \ref{algorithm_3} is an exact reachable set of the closed-loop system (\ref{closedloop}). However, the number of polytopes increases rapidly as time steps grow since the output of a neural network controller is a union of a number of polytopes, which makes the computational cost become intractable for  long intervals. To avoid the highly computational cost associated with the increasing number of polytopes for $\mathcal{X}_k$, we use convex hull of $\mathcal{X}_k$, that is using $\bar{\mathcal{X}}_k = \mathrm{conv}(\mathcal{X}_k)$ instead of $\mathcal{X}_k$ in Algorithm \ref{algorithm_3} to ensure that the output at each $k$ is only one polytope. However, the cost of this simplification is that the output is an over-approximation of the reachable set.

\begin{corollary}\label{cor3}
	Consider a  piecewise linear system of the form (\ref{system}) with initial set as in (\ref{inputset}) and a given safety specification $\mathcal{S}$, the closed-loop system given in (\ref{closedloop}) is safe for interval $[0,k]$ if $\bar{\mathcal{X}}_{[0,k]} \cap \neg \mathcal{S} = \emptyset$, where $\bar{\mathcal{X}}_{[0,k]}=\bigcup\nolimits_{h=0}^{k}\mathrm{conv}(\mathcal{X}_h)$.
\end{corollary}
\begin{proof}
	Since $\mathcal{X}_h\subseteq\mathrm{conv}(\mathcal{X}_h)$, we have reachable set $\mathcal{X}_{[0,k]}=\bigcup\nolimits_{h=0}^{k}(\mathcal{X}_h) \subseteq \bigcup\nolimits_{h=0}^{k}\mathrm{conv}(\mathcal{X}_h)=\bar{\mathcal{X}}_{[0,k]}$. Thus,  it is sufficient to say ${\mathcal{X}}_{[0,k]} \cap \neg \mathcal{S} = \emptyset$, if $\bar{\mathcal{X}}_{[0,k]} \cap \neg \mathcal{S} = \emptyset$. 
\end{proof}

\section{Numerical Example}
Let us consider a piecewise linear system with randomly generated system matrices as follows: 
\begin{align*}
A_1 = \left[\begin{array}{*{20}c}
-1.0609 & {  -1.0645}  \\
0.6600 & { -0.6178}  \\	
\end{array}\right],B_1 = \left[\begin{array}{*{20}c}
-0.9759  &  0.3688  \\
0.5874   & 2.5345  \\	
\end{array}\right]
\\
A_2 = \left[\begin{array}{*{20}c}
-0.5487 & {  -0.0196}  \\
0.3390 & { 1.2870}  \\	
\end{array}\right],B_2 = \left[\begin{array}{*{20}c}
0.5573 &   1.0926  \\
-0.6622 &   0.9284  \\	
\end{array}\right].
\end{align*}

The switching is assumed to be a periodic one defined as 
\begin{equation}
\sigma(k+1) = \left\{ {\begin{array}{*{20}c}
	1 & { \sigma(k) = 2}  \\
	2 & {\sigma(k) = 1}  \\	
	\end{array} } \right..
\end{equation}

Then, we let the input $u(k)$ be generated from a 2-layer ReLU neural network $g(x(k))$ with input $x(k)$. Assume there are 4 neurons in the hidden layer and 2 neurons in the output layer, and the weight matrices and bias vectors are randomly selected as  
\begin{align*}
W^{[1]} &= \left[\begin{array}{*{20}c}
 -0.4949 &  -0.4273 \\
 -0.6112  & -0.5277 \\
 -0.4287  & -0.5161\\
 0.0585  & -0.3319 \\	
\end{array}\right],
\theta^{[1]} = \left[\begin{array}{*{20}c}
-0.1971 \\
-0.2435 \\
0.9452\\
0.3945	\\
\end{array}\right]
\\
W^{[2]}& = \left[\begin{array}{*{20}c}
0.5891 &  -0.4770  &  0.0849 &   0.2686 \\
0.3210  & -0.2599  &  0.1594  & -0.0423	
\end{array}\right]
\\
\theta^{[2]}& = \left[\begin{array}{*{20}c}
-0.1862 \\
-0.1339
\end{array}\right].
\end{align*}

The initial set is given by $\mathcal{X}_0 = \{x \mid \left\|x\right\|_\infty\le 1,~x \in \mathbb{R}^{2}\}$. The estimated reachable sets based on Corollary \ref{cor3}, by which convex hull is used, for interval $[0,5]$ and $[0,10]$ are shown in Figs.\ref{fig2} and \ref{fig3}, which are depicted in green. We also discretize the initial set $\mathcal{X}_0$ by step 0.1 and generate 400 trajectories, which are all included in the estimated reachable set. 

With the output reach sets $\mathcal{X}_{[0,5]}$ and $\mathcal{X}_{[0,10]}$ in Figs. \ref{fig2} and \ref{fig3}, the safety property can be easily verified by inspecting the figures for non-empty intersections between the over-approximation of the reachable set and an unsafe region. For example, considering the unsafe region described by  $\neg \mathcal{S}= \{x \mid \left\|x-x_c\right\|_{\infty} \le 1,~x_c = [4,4]^{\top}\}$ which is depicted in red in Figs.\ref{fig2} and \ref{fig3},  it is easy to see that the closed-loop system is safe in interval $[0,5]$ since there is an empty intersection between reachable set $\mathcal{X}_{[0,5]}$ and unsafe set $\neg\mathcal{S}$, However, the safety property of closed loop system is uncertain over the time interval $[0,10]$ since $\mathcal{X}_{[0,10]} \cap \neg\mathcal{S} \ne \emptyset$. 

\begin{figure}
	\begin{center}
		\includegraphics[width=8.5cm]{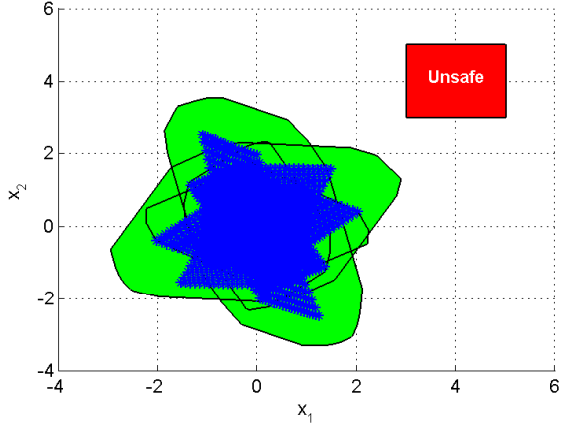}
		\caption{Estimated reachable set $\mathcal{X}_{[0,5]}$ is the green area. Blue markers * are 400 state trajectories from initial set $\mathcal{X}_0$. Red area is the unsafe region $\neg \mathcal{S}$. There is no intersection between $\mathcal{X}_{[0,5]}$ and $\neg \mathcal{S}$, thus the closed-loop system is safe in $[0,5]$. }
		\label{fig2}
	\end{center}
\end{figure}

\begin{figure}
	\begin{center}
		\includegraphics[width=8.5cm]{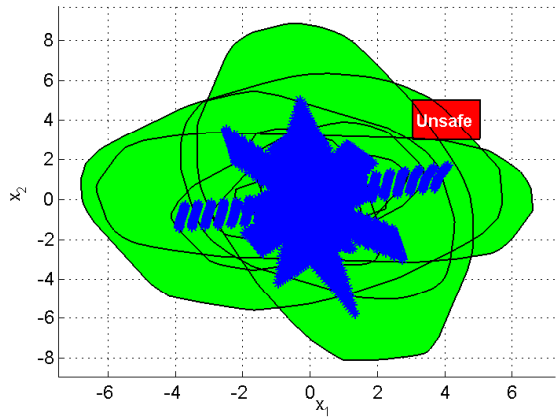}
		\caption{Green area is the estimated reachable set $\mathcal{X}_{[0,10]}$, blue * are 400 state trajectories and red area is the unsafe region $\neg \mathcal{S}$. The intersection between $\mathcal{X}_{[0,10]}$ and $\neg \mathcal{S}$ is not empty, thus the safety property of closed-loop system is uncertain for $[0,10]$, even though no simulated state trajectory enters the unsafe region.}
		\label{fig3}
	\end{center}
\end{figure}

\section{Conclusions}
The reachable set estimation problem for a class of discrete-time piecewise linear systems with neural network feedback controllers has been studied in this paper. First, a layer-by-layer computation method is proposed for computing neural networks consisting of ReLU neurons. The computation process is formulated as a set of polytope operations. An algorithm is proposed for reachable set estimation for piecewise linear systems with ReLU neural networks in a finite-time interval. Furthermore, the safety property of the closed-loop system can be verified by checking for nonempty intersections between the estimated output reachable set and unsafe regions. A numerical example is provided to show the effectiveness of the proposed approach.

\bibliographystyle{ieeetr}
\bibliography{ref}

\end{document}